\documentclass[10pt, twocolumn]{article}
\usepackage{latex8}
\usepackage{algorithmic}
\usepackage{algorithm}
\usepackage{epsfig}
\usepackage{subfigure}

\usepackage{times}
\usepackage{epsfig}
\usepackage{amssymb}
\usepackage{amsthm}
\usepackage{amsmath}
\usepackage{graphicx}
\usepackage{algorithm}
\usepackage{algorithmic}
\usepackage{setspace}
\usepackage{balance}

\renewcommand{\Section}[1]{\vspace{-3pt}\section{\hskip -1em.~~#1}\vspace{-5pt}} 
\renewcommand{\SubSection}[1]{\vspace{-3pt}\subsection{\hskip -1em.~~#1}
     	\vspace{-5pt}}
\newcommand{\SubSubSection}[1]{\vspace{-3pt}\subsubsection{\hskip -1em.~~#1}\vspace{-5pt}}

\title{Towards a Decentralized Algorithm for Mapping Network and Computational
Resources for Distributed Data-Flow Computations}

\author{\textit{Shah Asaduzzaman} and  \textit{Muthucumaru Maheswaran} \\
Advanced Networking Research Lab \\
School of Computer Science \\
McGill University \\
Montreal, QC H3A 2A7, Canada \\
\texttt{\{asad,maheswar\}@cs.mcgill.ca}
}

\begin{document}


\pagestyle{empty}
\maketitle
\thispagestyle{empty}
\begin{abstract}
  Several high-throughput distributed data-processing applications require
  multi-hop processing of streams of data. These applications include
  continual processing on data streams originating from a network of
  sensors, composing a multimedia stream through embedding several
  component streams originating from different locations, etc. These
  data-flow computing applications require multiple processing nodes
  interconnected according to the data-flow topology of the
  application, for on-stream processing of the data. Since the
  applications usually sustain for a long period, it is important to
  optimally map the component computations and communications on the
  nodes and links in the network, fulfilling the capacity constraints
  and optimizing some quality metric such as end-to-end latency. The
  mapping problem is unfortunately NP-complete and heuristics have
  been previously proposed to compute the approximate solution in a
  centralized way.  However, because of the dynamicity of the network,
  it is practically impossible to aggregate the correct state of the
  whole network in a single node. In this paper, we present a
  distributed algorithm for optimal mapping of the components of the
  data flow applications. We propose several heuristics to minimize
  the message complexity of the algorithm while maintaining the
  quality of the solution.
\end{abstract}

\Section{Introduction}
Real-time processing of continuous data streams are becoming an
important component of data-flow intensive distributed applications.
In general these applications consist of a few cascades of
computational operations on several streams of data originating from
one or more sources and presenting a view of the processed data at one
or more sink nodes. Applications such as continual
query~\cite{Schwan2005} on the stream of information sent by a network
of sensors, composing a multimedia stream through several stages of
encoding, decoding and embedding~\cite{Gu2006, Baochun2004},
scientific workflow~\cite{Kepler2006}, etc.\ belong to this category. These
applications require several computational resources along the path
the data streams travel from the source to destination. In addition,
as each of these computations generate new data streams that are to
processed by other computations or to be delivered to the
destination. Sufficient network link bandwidth must be provided to
carry these data streams among source, destination and computational
nodes, so that the computations can proceed seamlessly. In this paper,
we deal with the problem of optimally allocating computational and
network resources for these distributed applications.

Usually the distributed computation operates for a long time after
being set up with all the necessary resources. So, it is important to
optimally acquire the resources before the operation starts. When
resources are requested for a distributed job, the topology that
interconnect the component nodes of the flow, i.e.\ the data sources,
the processing nodes and the destination, is known. In very general
terms, the interconnection topology can be an acyclic graph. However,
in most common cases the flow is a linear path or tree or a
series-parallel graph. We show in Section~\ref{subsec:npcomplete} that
even for a linear path-like flow, finding a mapping that computations
on processing nodes and data transmissions on network paths,
satisfying the processing capacity and bandwidth constraint, is an
NP-complete problem. In this paper, we develop a scheme to solve the
problem of mapping linear path-like computation on an arbitrary
resource network.

The problem of establishing a path between a source and a destination
node in an arbitrary network, subject to some end-to-end quality
constraints, has been a topic for active research for a long time.  If
such path is to be established to satisfy one additive quality
requirement such as delay or hop-count, the problem can easily be
solved by Dijkstra's shortest path algorithm. Even if some end-to-end
min-max constraint such as bandwidth need to be satisfied, still the
problem can be solved easily using Wang and Crowcroft's
shortest-widest path algorithm~\cite{Crowcroft1996}. However, it is
well known that establishing a path satisfying more than one additive
quality constraints is an NP-hard problem~\cite{Nahrstedt1998,
Sahni2006}. It is important to note that the problem of finding a
mapping for a data-flow computation requires more than end-to-end
constraints, because computational capacity of each of the nodes need
to be individually satisfied.

Due to the inherent complexity of the optimization problem, several
workable heuristic solutions have been proposed in different contexts.
A recursive mapping on a hierarchy of node-groups in the resource
networks is applied in~\cite{Schwan2005}. In~\cite{Baochun2004}
and~\cite{Gu2006}, mapping is performed after pruning the whole
resource network into a subset of compatible resources. The solution
by Liang and Nahrstedt~\cite{Nahrstedt2006} is closest to ours. One of
the assumptions made by Liang and Nahrstedt was that the optimization
algorithm was executed in a single node and complete state of the
resource network is available to that node before execution. In a
large scale dynamic network this assumption is hard to realize. If we
assume that each node in the resource network is aware of the state of
its immediate neighborhood only, we need to compute the solution using
a distributed algorithm. In this paper we present a distributed
algorithm to solve the problem, which is a dynamic programming based
extension of the distributed Bellman-Ford algorithm.

The rest of the paper is organized as follows. In
Section~\ref{sec:problem} of this paper we formally define the
resource allocation problem as a constrained graph mapping
problem. The Bandwidth Constrained Path Mapping (BCPM) problem that covers
most of the practical applications, is then defined as a special case
of the general graph mapping problem. We provide a formal proof of
NP-completeness of the BCPM problem in the same section. In
Section~\ref{sec:algorithm}, centralized and decentralized algorithms
to solve the BCPM problem are developed. A guideline for designing
cost-effective heuristics to obtain approximate solutions to the
problem is provided at the end of the same section. The discussion is
then summarized with directions for possible future extensions in
Section~\ref{sec:conclusion}.

\Section{Problem Formulation}
\label{sec:problem}
In this section we formally define the problem of capacity constrained
mapping of dataflow computations on arbitrary networks. Any
distributed dataflow computation can be defined using three types of
nodes and interconnection between them. {\em Source nodes} are the
data sources originating the data streams. {\em Computing nodes} are
places where some computational operation on one or more incoming
data-stream is performed continually, and an output stream is
generated. {\em Sink nodes} are the places where the resulting flow
from the computation is presented. In a very general case, a dataflow
computation consists of one or more source nodes, one or more sink
nodes and zero or more computing nodes. The topology of data-flow
among these nodes is a directed acyclic graph (DAG). Although, theoretically
it is possible to have dataflow computations that have loops or
cycles, there will be finite number of iterations
of the data through the cycles and these iterations can be expanded
into finite acyclic graphs. In most common cases however, the dataflow
topology is a simple path consisting of a series of computing nodes, or
a tree where data-streams from multiple sources merged through several
steps and presented at a single sink. 

The network of computing and data-forwarding resources where the
distributed dataflow computation is to be instantiated can be
represented by an arbitrary graph. We denote this graph as resource
graph. Each node of the resource graph has a certain computational
capacity and each edge (link) of the resource graph has certain data
transmission capacity or bandwidth. In addition, each link may have
one or more additive quality metric, such as latency, jitter, etc.

\SubSection{Capacity Constrained Graph Mapping Problem}
In order to launch the distributed application on the network of
computers, we need to map the dataflow-DAG onto the resource graph such
that the computational and transmission requirements are fulfilled. If
there is more than one such feasible mapping, one would like to
choose the mapping that has minimum end-to-end delay on the resource
network.

More formally, we need to map a dataflow-DAG $G_J = (V_J,E_J)$ on to a
resource graph $G_R = (V_R, E_R)$. For each vertex $v_R \in V_R$, an
available computational capacity $C_{av}(v_R)$ is given. For each edge
$e_R \in E_R$, an available bandwidth $B_{av}(e_R)$ is given. In
addition, each edge $e_R \in E_R$ has an additive weight. For each
vertex $v_J \in V_J$, a computational requirement $C_{req}(v_J)$, and
for each edge $e_J \in E_J$, a bandwidth requirement $B_{req}(e_J)$ is
defined. There is a set of designated source nodes
$S_J \subset V_J = \{s_{1J},s_{2J}, ... ,s_{mJ}\}$ and a set of sink nodes $T_J
\subset V_J = \{t_{1J},t_{2J}, ... t_{nJ}\}$, such that $S_J \cap T_J = \phi$.

The bandwidth constrained DAG-mapping problem (BCDM) is to find a
mapping $M : V_J \rightarrow V_R$. For each source node $s_{iJ}$,
$M(s_{iJ}) = s_{iR}$ and for each sink node $t_{iJ}$, $M(t_{iJ}) =
t_{iR}$ are already given. It is important to note that multiple nodes
of the dataflow-DAG can map onto single node of the resource graph and
a single edge in the dataflow-DAG can span along a multi-hop path in
the resource graph.  So, defining the $V_J \rightarrow V_R$ mapping is
not sufficient to define the mapping of complete dataflow-DAG. In
addition to vertex mapping, another mapping $M_e: E_J \rightarrow P_R$ is
needed, where $P_R$ is the set of all possible paths in the resource
graphs, including zero length paths. Zero length paths are $(v,v)$
edges with infinite bandwidth and zero latency. Again, it is possible
that for two different edges, $e_1, e_2 \in E_J$, the mapped paths
$p_1 = M_e(e_1)$ and $p_2 = M_e(e_2)$ may have some common edges.

The mapping should fulfill the following constraints --

\begin{eqnarray}
\forall{v_R \in M(V_J)} &&\nonumber\\
\sum_{\{ v_J | v_J \in V_J, M(v_J) = v_R\}} {C_{req}(v_J)} 
     &\leq & C_{av}(v_R) \nonumber
\end{eqnarray} 
$$\forall {e_J = (u,v) \in E_J},$$ 
$$B(e_J) \leq min[ B(e_r), e_r \in M_e(e_J)]$$ 
We call this problem as Bandwidth Constrained DAG Mapping problem (BCDM).

When each edge $e_r \in E_r$ in the resource graph has an additive
metric $D(v_r)$, such as delay, cost, jitter, etc., we would like to find
the feasible mapping that minimizes the total cost
\begin{displaymath}D = \sum_{u,v\in V_j} \end{displaymath}

\begin{figure}[htbp]
\centering
\includegraphics{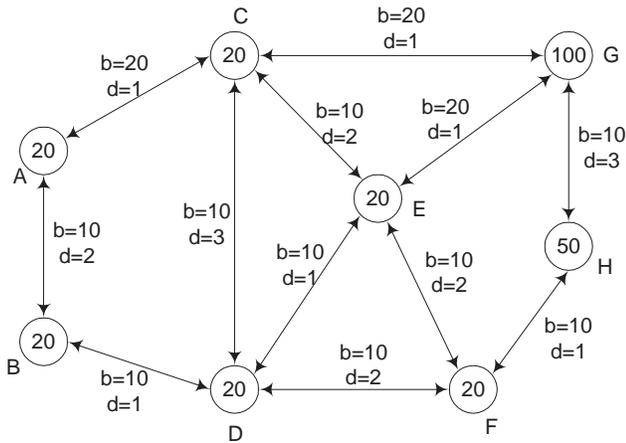}
\caption{An example resource network}
\label{fig:resource_graph}
\end{figure}

\begin{figure}[htbp]
\centering
\includegraphics{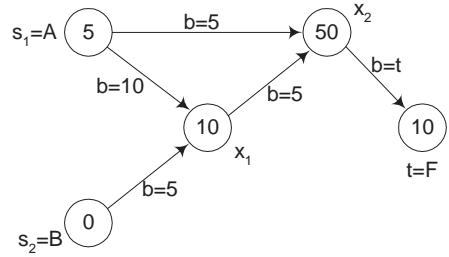}
\caption{An example data-flow computation with a DAG topology}
\label{fig:problem_dag}
\end{figure}

Figure~\ref{fig:resource_graph} shows an example resource network of
eight interconnected computing nodes. Computational capacity of each
node is represented by a number inside the node. The link bandwidth
and latency are mentioned on each edge. Figure~\ref{fig:problem_dag}
shows a dataflow-DAG containing $2$ source nodes $s_1$ and $s_2$, $2$
computing nodes $x_1$ and $x_2$, and one sink node $t$. $s_1$, $s_2$,
and $t$ must be mapped on resource node $A$, $B$, and $F$,
respectively. Each node in the dataflow-DAG has some processing
capacity requirement which is mentioned inside the node. Each link is
also annotated with a bandwidth requirement. A feasible mapping of
this dataflow-DAG on the resource graph is --

\begin{tabular}[h!]{c|c}
\begin{minipage}[t]{0.4\linewidth}
\begin{eqnarray*}
M(s_1) & = & A\\
M(s_2) & = & B\\
M(x_1) & = & E\\
M(x_2) & = & G\\
M(t)   & = & H\\
\end{eqnarray*}
\end{minipage}
&
\begin{minipage}[t]{0.48\linewidth}
\begin{eqnarray*}
M_e(s_1,x_1) & = & (A,C,E)\\
M_e(s_2,x_1) & = & (B,D,E)\\
M_e(x_1,x_2) & = & (E,G)\\
M_e(s_1,x_2) & = & (A,C,G)\\
M_e(x_2,t)   & = & (G,H,F)
\end{eqnarray*}
\end{minipage}
\end{tabular}

\SubSection{Constrained Path Mapping Problem}
Although in very general terms the dataflow computation resembles a
DAG topology, in most practical cases the topology is a simple path.
Given that the mapping of a DAG efficiently on the resource network
with all the constraints satisfied is hard to solve, it is useful to to
tackle the simpler problem of bandwidth constrained path mapping
problem (BCPM) first. In BCPM, the topology of the data flow
computation is restricted to a directed loop-free path, with a single
source and a single sink.

Precisely, we are given a dataflow path $P_J = (V_J, E_J)$, $V_J =
{v_0=s, v_1, v_2, ..., v_m = t}$ and $E_J = \{e_i = (v_i, v_{i+1}) | 0
\leq i < m \}$ to map on the resource graph $G_R = (V_R, E_R)$ defined
in the previous section. Each node $v_i, 0 \leq i \leq m$ of the
program path has a computational capacity requirement $C_{req}(v_i)$,
and each edge $e_i = (v_i, v_{i+1}), 0 \leq i < m$ has a bandwidth
requirement $B_{req}(e_i)$. We need to find the mappings $M: V_J
\rightarrow V_R$ and $M_e:E_J\rightarrow E_R$ that satisfies the
constraints. Mapping of $s$ and $t$ is already given.

\begin{figure}[htbp]
\centering
\includegraphics{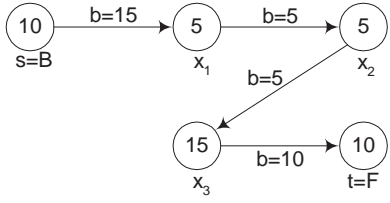}
\caption{An example data-flow computation with a path topology}
\label{fig:problem_path}
\end{figure}

An example dataflow path with one source $s$, one sink $t$ and three
computational nodes $x_1$, $x_2$, $x_3$ is shown in
Figure~\ref{fig:problem_path}, with the node capacity and bandwidth
requirements. $s$ and $t$ must be mapped on $B$ and $F$,
respectively. There can be many feasible mappings of this dataflow
computation on the resource graph in
Figure~\ref{fig:resource_graph}. One of them is --

\begin{tabular}[h!]{c|c}
\begin{minipage}[t]{0.4\linewidth}
\begin{eqnarray*}
M(s) & = & B\\
M(x_1) & = & B\\
M(x_2) & = & B\\
M(x_3) & = & D\\
M(t)   & = & F\\
\end{eqnarray*}
\end{minipage}
&
\begin{minipage}[t]{0.48\linewidth}
\begin{eqnarray*}
M_e(s,x_1) & = & (B,B)\\
M_e(x_1,x_2) & = & (B,B)\\
M_e(x_2,x_3) & = & (B,D)\\
M_e(x_3,t) & = & (D,F)
\end{eqnarray*}
\end{minipage}
\end{tabular}
which is also optimal in terms of total end-to-end latency of the
resource nodes $M(s)$ and $M(t)$.

\SubSection{Computational Complexity of the Problem}
\label{subsec:npcomplete}
We will now prove that BCPM problem is NP-complete. Since, BCPM is a
special case of BCDM, NP-completeness of BCPM iplies that BCDM is an
NP-hard problem. The NP-completeness proof of the BCPM problem is
constructed by transformation of the Longest Path
problem~\cite{Garey1979}. Definition of the decision version of the
Longest Path problem is as follows -

{\em Instance}: A graph $G = (V, E)$, a length function $l:E\rightarrow Z^+$,
specified vertices $s,t \in V$ and a positive integer $K$. 
{\em Question}: Is there an $(s \leadsto t)$ simple path $P \subseteq
G$ such that $\sum_{e \in P}{l(e) \geq K}$ ?

It is known that Longest Path problem is NP-complete, even for a
special case, where $\forall_{e \in E}l(e) = 1$~\cite{Garey1979}. We
will show that any instance of this special Longest Path problem can
be polynomially transformed into an instance of BCPM.

\SubSubSection{Longest Path $\propto$ BCPM}
We construct an instance of BCPM as follows - 

We take $G_R(V_R,E_R) = G(V,E)$, $\forall_{v \in V_R} C_{av}(v) = 1$,
$\forall_{e \in E_R} B_{av}(e) = 1$. Take a simple path $P_J = (V_J,
E_J)$ such that $|V_J| = K$, $\forall_{v \in V_J} C_{req}(v) = 1$ and
$\forall_{e \in E_J} B_{req}(e) = 1$. 

Now, if there is a simple $(s\leadsto t)$ path of length $\geq K$ in
G, then that path must have K hops, since $\forall_{e \in E}l(e) =
1$. Therefore, we can map $P_J$ along the corresponding path
$P_J\prime$ in $G_R$. If $|P_J\prime| > K$, then we can map first
$K-1$ nodes of $P_J$ on $P_J\prime$ and map the remaining edge
$u_{K-1}, u_K$ on the $v \leadsto t$ subpath of $P_J\prime$, where
$u_{K-1}$ is mapped on $v$. 

Given a mapping of the path $P_J$ on a path $P_J\prime \subseteq G_R$
that satisfies the capacity and bandwidth requirement constraints,
$|P_J\prime|$ must be $>=K$, because no two vertices of $P_J$ can be
mapped on a single vertex of $|P_J\prime|$ given the abovementioned
capacity constraints.

\SubSubSection{BCPM $\in$ NP}
Given an arbitrary mapping $M:V_J \rightarrow V_R$ one can
polynomially verify - 
\begin{itemize} 
\item Whether $C_{req}(v) \leq C_{av}(M(v))$, for all $v \in V_J$.
\item For each edge $(u,v) \in P_J$, whether there is a $(M(u)
  \leadsto M(v))$ path in $G_R$ that satisfies the bandwidth constraint
  of $(u,v)$ (Similar to bandwidth constrained shortest path
  problem~\cite{Crowcroft1996}).
\end{itemize}

This completes the proof that $BCPM \in NP$-$C$.

\Section{Algorithm for path mapping problem}
\label{sec:algorithm}
To solve the BCPM problem, we developed an algorithm using the the
Bellman-Ford relaxation scheme. First, we present the centralized
version of the algorithm, where the whole mapping is computed by a
single node that has knowledge of the state of the whole network of
nodes. Later, we explain the development of the distributed algorithm
based on this centralized one.

This algorithm works by relaxing along each edge of the resource graph
$N-1$ times, where $N = |V_R|$, the number of nodes in the resource
graph. For each node $u$ of the resource graph, a set of feasible
mappings of different length prefixes of the dataflow-path on any
resource path from the source node $s$ to the current node, is
maintained. In each relaxation along an $(u,v)$ edge, any new feasible
map on $(s \leadsto u)$ is extended in all possible ways, to complete
the list of feasible maps of dataflow path-prefixes on the resource
path $(s \leadsto u, v)$ and these new partial mappings are added to
the set maintained for node $v$. After $N-1$ iterations of relaxation
of all edges, the map set maintained for terminal node $t$ contains
all the feasible mappings of the dataflow-path on any $(s\leadsto
t)$ resource path. The algorithm is presented in
Algorithm~\ref{alg:pathmap_central_main},
\ref{alg:pathmap_central_relax}
and~\ref{alg:pathmap_central_extend}. A formal proof of the
correctness of the algorithm is presented in the following
sub-section. Lines $10$-$12$ of the subroutine Relax is added to
terminate the algorithm as soon as one feasible $(s\leadsto t)$
mapping is found. These lines should be omitted when optimal mapping
is sought.

We have computed the computational complexity of the algorithm in
Section~\ref{subsec:complexity}. The complexity is bounded by
polynomial of the size of the partial map set $S$, although the set
size is exponential.  The problem being NP-hard, it is impossible to
have a polynomially bounded optimal algorithm. However, heuristics may
be applied to produce sub-optimal solutions within a tractable amount
of complexity.  A good way of designing such heuristics is to restrict
the size of the map-set in some way. In Section~\ref{subsec:heuristic}
we have discussed several possible heuristics to solve the BCPM
problem. Note that because the set of partial map is stored in each
node, the memory complexity of the algorithm becomes exponential
too. This can be avoided by omitting the storage of partial maps. Each
partial map need to be stored for one iteration of relaxation only. If
partial maps are deleted after relaxation, the set size never grows
beyond $O(dp)$, where, $d$ is the average indegree of a node in
resource graph and $p = |P_J|$ is the number of nodes in the dataflow
path.

\begin{algorithm}[h]
     \begin{algorithmic}[1]
     \FOR{$x=0$ to $|P_J|-1$}
       \IF{$\sum_{0 \leq k \leq x}{C_{req}(k)} \leq C_{av}(s)$}
         \STATE{$M(s,x) = \{m|m$ maps initial $x$ nodes of $P_J$ on $s\}$} 
       \ELSE
         \STATE{\textbf{break}}
       \ENDIF
     \ENDFOR

     \FOR{each vertex $v \in V_R-{s}$}
       \FOR{$i=0$ to $|P_J|$}
         \STATE{$M(v,i) = \phi$}
       \ENDFOR
     \ENDFOR

     \FOR{$i = 1$ to $|V_R| - 1$}
       \FOR{each edge $ e = (u,v) \in E_R $}
         \STATE{Relax(u,v)} 
       \ENDFOR
     \ENDFOR

    \end{algorithmic}
     \caption{\em Pathmap($P_J$, $G_R$)}
  \label{alg:pathmap_central_main}
\end{algorithm}

\begin{algorithm}[h]
    \begin{algorithmic}[1]
     \FOR{$j = 0$ to $|P_J|$}
       \STATE{$M_{tmp}(j) = null$}
     \ENDFOR

     \FOR{$j = 0$ to $|P_J|-1$}
       \IF {$B_{req}(j, j+1) \leq B_{av}(u,v)$} 
         \FOR{each new mapping $m \in M(u,j)$ in the last iteration}
           \IF{$v == t$}
	     \STATE{$m_x = $ Extend($m$, j, $|P_J|-j$, v)}
	     \STATE{$M(v,|P_J|) = M(v,|P_J|) \cup m_x$}
	     \IF{$M(v,|P_J|) \neq \phi$}
	        \STATE{terminate the algorithm with $M(v,|P|)$ as
                   result}
	     \ENDIF	
	   \ELSE
	     \FOR{$x=0$ to $|P_J|-j-1$} 
	       \STATE{$m_x = $ Extend(m, j, x, v)}
	       \IF{$m_x \neq null$}
	         \STATE{$M(v,j+x) = M(v,j+x) \cup m_x$}
	       \ELSE
	         \STATE{\textbf{break}}    
	       \ENDIF
             \ENDFOR
	   \ENDIF
	   \STATE{mark $m$ as old}
         \ENDFOR
       \ENDIF	 
      \ENDFOR 
    \end{algorithmic}
  \caption{\textbf{subroutine} {\em Relax(u,v)}}
  \label{alg:pathmap_central_relax}
\end{algorithm}

\begin{algorithm}[h]
    \begin{algorithmic}[1]
     \IF{$\sum_{1 \leq k \leq x}{C_{req}(j+k)} \leq C_{av}(v)$}
       \STATE{extend $m$ by putting computations $\{j+1, j+2,
     ..., j+x\}$ in node $v$} 
       \STATE{let $m_x$ be the extended mapping} 
     \ELSE
       \STATE{$m_x = null$}
     \ENDIF
     \RETURN{$m_x$}
     \end{algorithmic}  
    \caption{\textbf{subroutine} {\em Extend(m, j, x, v)}}
    \label{alg:pathmap_central_extend}
\end{algorithm}

\SubSection{Correctness of BCPM algorithm}
\label{proof_correctness}
In this section we give a formal proof that when BCPM algorithm
terminates, $M(t,|P_J|)$ always contains a feasible mapping of $P_J$ on
$G_R$ if and only if such a feasible mapping exists.

\newtheorem{thm}{Theorem}[section]
\newtheorem{lem}[thm]{Lemma}
\newtheorem{cor}[thm]{Corollary}

\begin{lem}
\label{lemma_1}
If $M(u) = \bigcup_{\forall{j}} M(u,j)$ contains all feasible mappings
of different length prefixes of $P_J$ on an path $(s \leadsto u) \in
G_R$, then after computing $\textbf{Relax}(u,v)$, $M(v)$ includes all
feasible mappings of different length prefixes of $P_J$ on the path $(s
\leadsto u, v) \in G_R$.
\end{lem}

\begin{proof}
By the construction of the $\textbf{Relax}(u,v)$
subroutine, each mapping $m \in M(u,j)$, of a $j$-length prefix of
$P_J$ on a $(s \leadsto u)$ path, is extended over the $(u,v)$ edge
exactly once. Any possible mapping of a $k$-length prefix of $P_J$ on
the $(s \leadsto u, v)$ path can be divided into 2 sub-mappings: a
mapping of $j$-length prefix $(j \leq k)$ of $P_J$ on $(s \leadsto u)$
path and a mapping of the following $k-j$ vertices of the $k$-length
prefix on $v$. Since all feasible sub-mappings of the first kind is
included in $M(u)$ and all the extensions of the second kind is
considered in lines $8$ to $14$ and $15$ to $22$ of
$\textbf{Relax}(u,v)$, $M(v)$ contains all feasible mappings of any
prefix of $P_J$ on $(s \leadsto u, v)$ paths.
\end{proof}

\begin{lem}
\label{lemma_2}
For any node $v \in V_R$ if there is a $s \leadsto v$ path $(v_0 = s,
v_1, v_2, ..., v_k = v)$ of length $k$, after $k$th iteration of the
outer for loop in line $7$ of the $PathMap$ algorithm, all feasible
mappings of different length prefixes of $P_J$ on the $(v_0 \leadsto
v_k)$ path has been recorded in $M(v)$.
\end{lem}

\begin{proof}
We will prove by induction on $k$. When $k=0$, i.e.\
after the initialization phase, $M(v_0,i)$ or $M(s,i), 0 \leq i \leq
|P_J|$ contains the feasible $i$-length prefix with first $i$ vertices
of $P$ mapped on $s$. So the basis is true.

Now let us assume that after $i-1$ iterations, $0 < i \leq k$,
$M(v_{i-1})$ contains all feasible mappings of different lengths on
the $(s\leadsto v_{i-1})$ portion of the $(s\leadsto v_k)$ path. Since
each edge in $E_R$ is considered once in each iteration,
$Relax(v_{i-1}, v_i)$ must be called in the $i$th iteration too.
So, by Lemma~\ref{lemma_1}, we can conclude that all feasible prefix
mappings of $P_J$ on the $(s \leadsto v_i)$ path is included in
$M(v_i)$.  
\end{proof}

\begin{thm}
\label{thm_3}
After $|V_R|-1$ iterations of the outer loop in line $7$ algorithm
$Pathmap$, for each node $v \in V_R$, $M(v)$ contains all feasible
mappings of different length prefixes of $P_J$ on all possible $s
\leadsto v$ paths.
\end{thm}

\begin{proof}
Since there is no simple path longer than $|V_R| - 1$,
according to Lemma~\ref{lemma_2}, all such paths will be covered by
the $Relax$ procedure after $|V_R|-1$ iterations.
\end{proof}

The fact that after termination of $Pathmap$, $M(t)$ contains all the
feasible maps of $P_J$ on possible $(s \leadsto t)$ paths, follows
directly from Theorem~\ref{thm_3} with inclusion of lines $7$ to $12$ in
the {\em Relax} procedure.

\SubSection{Complexity of the algorithm}
\label{subsec:complexity}
The problem size parameters are $|V_R| \equiv n$, $|E_R| \equiv e$ and
$|P_J| \equiv p$. The outer loop of {\em Pathmap} is iterated $n-1$
times and each iteration considers each of the $e$ edges exactly
once. So, the {\em Relax} procedure is called $ne$ times. In each
relaxation over an edge $(u,v)$, each of the $p$ prefix mappings from
$M(u)$ is tried for relaxation into some of the $p$ mappings in
$M(v)$. A $j$ length prefix in $M(u,j)$ is tried for relaxation into
$p-j$ of the $M(u,i), j\leq i \leq p$, and each trial requires $(i-j)$
computations of constant complexity for the extension. Let $S$ be the
maximum number of entries in the set of mappings $M(u,j), u \in V_R,
0\leq j \leq p$. Note that only the new entries are relaxed in each
iteration. However, the upper bound on the number of entries relaxed
per $M(u,j)$ will be $S$. So, the complexity of {\em Relax(u,v)} is --

\begin{eqnarray}
S\sum_{j=0}^{p-1}{\left( \sum_{x=1}^{p-j-1}{x} + 1 \right) } & = &
S\left(\frac{5}{12}p^3 + \frac{1}{4}p^2 + \frac{2}{3}p\right)
\nonumber \\
& = & O\left(\frac{5}{12}p^3{S}\right) \nonumber
\end{eqnarray}

So, the overall time complexity of the algorithm becomes
$O(nep^3S)$. We see that the sets $M(u,j)$ are creating the major load
on both time and memory complexity of the algorithm. Therefore,
restricting the growth of $S$ within polynomial limit would possibly
result in a polynomial time approximation algorithm.

\SubSection{Distributed version of the algorithm}
The centralized algorithm can be easily extended to a distributed
version, where each node $u$ in the resource network $G_R$ will
maintain the data structure $M(u)$ of partially computed
mappings. Also, node $u$ will be responsible for computing the
relaxation to each of its neighbors $v$ in $G_R$. The extended
mappings are then transmitted to $v$. The relaxation procedure is
invoked by a node $u$ when any new mapping arrives from any of its
incoming neighbors. The algorithm is formally laid out in
Algorithm~\ref{alg:pathmap_distr}. Upon arrival of a map message $m$,
a node $u$ process the message using the algorithm {\em ProcessMap(u,
m)}.  It follows from the correctness of the centralized algorithm
that the distributed mapping completes after at most $N-1$ {\em
ProcessMap} invocation by each node in the graph. The distributed
mapping algorithm can be terminated by force as soon as the terminal
nodes receives a complete mapping. Otherwise, the algorithm terminates
after all the outstanding {\em ProcessMap} have been completed. Since
cycles are avoided during extension, an initial mapping may be
extended at most $N-1$ times. Thus there will be a finite number of
{\em ProcessMap} invocation and the algorithm will terminate after a
finite amount of time.

\begin{algorithm}[h]
    \begin{algorithmic}[1]
       \STATE{Map message contains the mapping of computation nodes
       $0$,$1$,$2$, ... , $j$ on resource nodes. The first message to a
       node contains the requirement definition of the computation too}  
       \STATE{$j = |m|$}
       \IF{$u == t$}
          \STATE{$m_x = $ Extend($m$, j, $|P_J|-j$, u)}
	  \IF{$m_x \neq null$}
	     \STATE{terminate the algorithm with $m_x$ as
                   result}
	  \ENDIF	
       \ELSE
          \FOR{$x=0$ to $|P_J|-j-1$} 
	     \STATE{$m_x = $ Extend(m, j, x, u)}
	     \IF{$m_x \neq null$}
	          \FOR {each neighbor $v$ of $u$ that is not already in $m$}
	             \IF {$B_{req}(j+x, j+x+1) \leq B_{av}(u,v)$} 
	                 \STATE{extend $m_x$ to $m_xx$ by appending a map
       of $0$ computations on node $v$}
		         \STATE{{\bf send} $m_xx$ to $v$}
		     \ENDIF
		  \ENDFOR   
	     \ENDIF
          \ENDFOR
       \ENDIF
   \end{algorithmic}
    \caption{\em ProcessMap(u, m)}
    \label{alg:pathmap_distr}
\end{algorithm}

\SubSection{Heuristic Approaches to Reduce Complexity}
\label{subsec:heuristic}
Computational complexity of both the centralized and the distributed
path mapping algorithm grows exponentially with the problem
size. Therefore, for practical deployment, we need some heuristic that
produces good approximation to the optimal result. Here we discuss
three possible heuristics that modifies the original algorithm to
reduce computational, messaging and memory complexity.

\SubSubSection{LeastCostMap}
One major source of growth in complexity of the algorithm is the
exponential growth of the set of partial maps maintained for each
node. In the {\em LeastCostMap} heuristic, only one partial map of
each prefix-length is maintained for each node. If a new map is generated,
the cost of the new map in terms of the additive quality metric is
compared with that of the already stored one, and the map with higher
cost is discarded. This policy reduces the complexity to $O(p^3)$.

Similar policy can be applied to the distributed version of the
algorithm. However, in the distributed case, a map message is expanded
to its neighbors as soon as the message is received. So, if a higher
cost map message is arrived before a lower cost one, the processing of
the higher cost message cannot be pruned. However, in  most cases,
higher cost messages arrive later, so they are pruned.

We have implemented both the centralized and distributed version of
the original algorithm and also the LeastCostMap heuristic. The
algorithms are then applied on random topologies generated by the
BRITE Internet topology generator~\cite{Brite2001} and randomly
generated dataflow paths. Due to the huge computational complexity of
the exact algorithm, it was not possible to run it for networks larger
than $50$ nodes. For these networks, the heuristic is able to find the
optimal solution in $99\%$ of the cases, with $100$ to $1000$ fold
reduction in the size of the set of partial maps. For similar
topologies, the distributed version of the heuristic produced optimal
result in more than $99\%$ cases and total number of message exchange
was reduced approximately $100$ fold.
  
\SubSubSection{AnnealedLeastCostMap}
One way of trading off between optimality and complexity of the {\em
LeastCostMap} heuristic is to apply a simulated annealing approach
to decide whether to discard a higher cost partial map from the set
in presence of a lower cost map. As the temperature of the
process anneals, i.e. at the later iterations, the probability of
keeping a non-minimal partial solution will decrease. Definitely this
approach increases the computation and message complexity. However,
this allows some of the non-minimal partial solutions to grow and
possibly lead to a better complete solution.

\SubSubSection{RandomNeighbor}
Another way of restricting the message complexity is to extend any
partial map to a randomly chosen subset of $k$ neighbors instead of
expanding to all of them. Higher values of $k$ increases the chance of
getting the optimal solution. The {\em RandomNeighbor} heuristic with
$k=1$ did not produce results as good as LeastCostMap, although number
of messages were reduced dramatically. Further investigation need to
be done to determine a suitable value of $k$.

\Section{Conclusion}
\label{sec:conclusion}
In this paper we have developed and explained a decentralized
algorithm to compute the optimal mapping of computational capacity and
network bandwidth requirement of a data-flow computation. Many
high-throughput scientific research platforms need to support
applications that resemble data-flow computation. The discussion
presented in this paper provides in-depth understanding of the
resource allocation problem for such computations and demonstrates the
way to develop cost-effective solutions. At this point, the algorithm
supports computations with path-topology only. Several interesting
applications such as complex continual queries on data stream
originating from multiple sites, resemble a tree topology. A possible
extension of this work is to modify the algorithm such that mapping of
flow-computations with different topologies can be obtained.

\singlespace
\vspace{-.1cm}
\balance
\bibliographystyle{latex8}
\bibliography{hpcs}

\end{document}